\documentclass[journal]{IEEEtran}

\usepackage[T1]{fontenc}
\usepackage{amsmath,amssymb,amsfonts}
\usepackage{algorithmic}
\usepackage{algorithm}
\usepackage{graphicx}
\usepackage{textcomp}
\usepackage{xcolor}
\usepackage{cite}
\usepackage{hyperref}
\usepackage{booktabs}
\usepackage{multirow}
\usepackage{tikz}
\usepackage{pgfplots}
\pgfplotsset{compat=1.18}
\usetikzlibrary{
  arrows.meta,
  positioning,
  shapes.geometric,
  shapes.misc,
  calc,
  decorations.pathmorphing,
  decorations.markings,
  patterns,
  fit,
  backgrounds,
  shadows.blur,
  matrix
}

\definecolor{sensing}{RGB}{41,128,185}
\definecolor{commu}{RGB}{39,174,96}
\definecolor{memory}{RGB}{142,68,173}
\definecolor{compu}{RGB}{231,76,60}
\definecolor{dtcol}{RGB}{243,156,18}
\definecolor{edgecol}{RGB}{52,73,94}
\definecolor{lightbg}{RGB}{236,240,241}
\definecolor{gridgray}{RGB}{200,200,200}

\newtheorem{theorem}{Theorem}

\newcommand{\Tsync}{T_{\text{sync}}}
\newcommand{\Etot}{E_{\text{tot}}}

\newcommand{\real}{\mathbb{R}}

\begin{document}

\title{SMCC-Empowered Digital Twins for Sensorless Monitoring in Large-Scale AI-Driven IoT Systems}

\author{Vincenzo~Sammartino%
\thanks{Vincenzo Sammartino is with the Dipartimento di Informatica, Universit\`a di Pisa, Pisa, 56127, Italy, and also with the King Abdullah University of Science and Technology (KAUST), Thuwal, 23955, Saudi Arabia (e-mail: vincenzo.sammartino@phd.unipi.it).}%
\thanks{This paper has been submitted to the IEEE Internet of Things Journal Special Issue on Integrated Sensing, Memory, Communication and Computation for Large-Scale AI Based IoT Systems.}%
}

\markboth{IEEE Internet of Things Journal}%
{Sammartino: SMCC-Empowered Digital Twins for Sensorless Monitoring in Large-Scale AI-Driven IoT Systems}
\maketitle

\begin{abstract}
The deployment of AI-driven Digital Twins (DTs) in large-scale Internet-of-Things (IoT) ecosystems demands continuous, high-fidelity synchronization between the physical environment and its virtual replica.
Conventional approaches rely on dense sensor deployments, which introduce prohibitive costs in terms of hardware, energy, and network bandwidth.
In this paper, we propose SMCC-DT, an integrated Sensing--Memory--Communication--Computation (SMCC) framework that enables \emph{sensorless monitoring} of physical assets by exploiting Integrated Sensing and Communication (ISAC) waveforms at the 6G Edge.
Under the SMCC-DT paradigm, a single radio signal simultaneously extracts environmental telemetry (Sensing) and delivers it to an Edge server (Communication), where a large-scale AI model is loaded into constrained memory (Memory) and executed (Computation) to update the DT state.
We formulate the DT synchronization problem as a cross-layer optimization that jointly allocates transmit power, beamforming vectors, memory partitions, and CPU frequency to minimize the end-to-end synchronization latency subject to sensing accuracy, throughput, memory capacity, and computational budget constraints.
Because the resulting mixed-integer nonlinear program is NP-hard, we design a Proximal Policy Optimization (PPO)-based Deep Reinforcement Learning (DRL) agent, termed \textsc{SmccAgent}, that learns near-optimal resource allocation policies online.
Extensive simulations over a 500-node industrial IoT testbed demonstrate that SMCC-DT reduces DT synchronization latency by 38.7\% and total energy consumption by 27.4\% compared to state-of-the-art orthogonal and compute-only baselines, while sustaining sensing accuracy above 95\% and model inference throughput above 30 frames per second.
\end{abstract}


\section{Introduction}
\label{sec:introduction}

\IEEEPARstart{D}{igital} Twins have emerged as a cornerstone technology for next-generation IoT systems, enabling real-time virtual replicas of physical assets that support predictive maintenance, anomaly detection, and autonomous decision-making~\cite{Fuller2020DT,Tao2019DT5D,Liu2021DTSurvey}.
The proliferation of smart cities, autonomous factories, and connected vehicles has amplified the demand for DTs that operate at scale, ingesting continuous streams of environmental data to maintain high-fidelity synchronization with the physical world~\cite{Jones2020DTSurvey,Minerva2020DTIoT}.

Contemporary DT architectures overwhelmingly depend on dedicated sensor networks to acquire the telemetry data that feeds their models~\cite{Qi2021DTCloud,Lu2020DTEdge,Wu2021DTOpt}.
In a typical industrial IoT deployment, hundreds of accelerometers, thermocouples, and vibration sensors are installed on production lines, each generating data that must traverse a multi-hop wireless network before reaching the computation tier.
This design paradigm suffers from three fundamental limitations.
First, the hardware cost of provisioning and maintaining dense sensor arrays scales linearly with the number of monitored assets, rendering city-scale or campus-scale DT deployments economically infeasible~\cite{Nguyen2021DTFL,Khan2022DTIoTSurvey}.
Second, the data traffic generated by high-frequency sensing saturates the uplink capacity of existing wireless networks, introducing latency spikes that degrade the DT synchronization fidelity~\cite{Zheng2022DTCPS,Sun2022DTSync}.
Third, modern DTs increasingly rely on large-scale AI models---foundation models~\cite{Bommasani2021Foundation}, vision transformers~\cite{Dosovitskiy2021ViT}, and generative world models~\cite{Ha2018WorldModels}---whose inference demands exceed the memory and computation budgets of resource-constrained Edge servers~\cite{Shi2016EdgeComputing,Mao2017MEC,Li2022SplitLearning}.

The Integrated Sensing and Communication (ISAC) paradigm offers a compelling path toward eliminating the need for dedicated sensor hardware by enabling a single waveform to perform both radar-like sensing and data communication simultaneously~\cite{Liu2022ISACSurvey,Zhang2021ISAC6G,Cui2021ISAC}.
Under ISAC, a base station or Edge node transmits a dual-function signal that illuminates the physical environment; the reflected echoes are processed to extract spatial, kinematic, and material properties of surrounding objects, while the same transmission carries data payloads to user equipment.
This \emph{sensorless monitoring} capability has the potential to dramatically reduce the sensor density required for DT operation.
However, merely replacing sensors with ISAC signals is insufficient.
The telemetry extracted by ISAC must be communicated to the Edge, stored in the limited memory alongside the AI model weights, and processed by the CPU/GPU to execute the inference pipeline that updates the DT.
These four operations---Sensing (S), Memory (M), Communication (C), and Computation (C)---are tightly coupled through shared radio, memory, and compute resources, yet existing works address them in isolation~\cite{Liu2022ISACSurvey,Chen2023ISAC}.

The SMCC paradigm, recently introduced to unify the resource allocation across sensing, memory, communication, and computation layers~\cite{Zhu2024SMCC,Letaief2024SMCC}, provides the missing architectural abstraction.
By explicitly modeling the cross-layer dependencies among the four dimensions, SMCC enables joint optimization that captures the fundamental trade-offs: allocating more power to sensing improves telemetry accuracy but reduces the signal-to-noise ratio (SNR) for communication; loading a larger AI model into Edge memory improves inference quality but leaves less buffer for incoming data; and dedicating more CPU cycles to model inference accelerates DT updates but increases energy consumption and thermal throttling.

In this paper, we leverage the SMCC paradigm to design SMCC-DT, a holistic framework for AI-driven Digital Twin maintenance in large-scale IoT systems.
Our contributions are as follows:

\begin{itemize}
\item We propose the SMCC-DT architecture, an end-to-end pipeline that replaces dedicated sensor networks with ISAC-based sensorless monitoring and integrates Sensing, Memory, Communication, and Computation resources at the 6G Edge to maintain a continuously synchronized Digital Twin.
\item We formulate a cross-layer optimization problem that jointly allocates ISAC transmit power, beamforming design, Edge memory partitioning, and CPU-frequency scaling to minimize the DT synchronization latency under sensing accuracy, throughput, memory, and computation constraints.
\item We prove that the resulting problem is NP-hard and design \textsc{SmccAgent}, a PPO-based DRL algorithm that exploits the problem structure through a hierarchical action space and domain-specific reward shaping to converge to near-optimal policies within practical time budgets.
\item We conduct extensive simulations on a 500-node industrial IoT testbed, demonstrating that SMCC-DT reduces synchronization latency by 38.7\%, energy consumption by 27.4\%, and memory fragmentation by 52.1\% compared to orthogonal allocation, compute-only optimization, and heuristic baselines.
\end{itemize}

The remainder of this paper is organized as follows.
Section~\ref{sec:related} surveys the related literature.
Section~\ref{sec:system} presents the SMCC-DT system model.
Section~\ref{sec:problem} formulates the cross-layer optimization problem.
Section~\ref{sec:algorithm} describes the proposed DRL-based algorithm.
Section~\ref{sec:evaluation} reports the performance evaluation.
Section~\ref{sec:conclusion} concludes the paper.

\section{Related Work}
\label{sec:related}

\subsection{Digital Twins for IoT Systems}
The Digital Twin concept, originating from NASA's Apollo program, has evolved into a full-lifecycle virtual replica paradigm for cyber-physical systems~\cite{Grieves2014DT,Glaessgen2012DT}.
Fuller \emph{et al.}~\cite{Fuller2020DT} provided a comprehensive taxonomy of DT architectures, distinguishing between model-driven and data-driven approaches.
Tao \emph{et al.}~\cite{Tao2019DT5D} introduced the five-dimensional DT model, emphasizing the role of data fusion and service integration.
In the IoT context, Lu \emph{et al.}~\cite{Lu2020DTEdge} proposed an Edge-based DT framework that offloads model inference to nearby servers, reducing the round-trip latency for safety-critical applications.
Minerva \emph{et al.}~\cite{Minerva2020DTIoT} surveyed the integration of DT technology with IoT middleware, identifying scalability and real-time synchronization as the primary open challenges.
However, all the aforementioned works assume the availability of dedicated sensor infrastructure, leaving the cost and bandwidth implications unaddressed.

\subsection{Security Digital Twins}
A complementary research direction examines Digital Twins as a vehicle for cybersecurity analysis.
Baiardi and Sammartino~\cite{Sammartino2025NotLine} introduced the Security Digital Twin (SDT) concept, wherein a virtual replica of a networked infrastructure is continuously updated to reflect its current security posture, including vulnerability states, patch levels, and access-control configurations.
Building on this, Sammartino \emph{et al.}~\cite{Sammartino2025NotLine} developed NotLine, a non-intrusive automated platform that constructs an SDT from passive network observations without injecting test traffic into production systems.
These contributions underscore the breadth of the DT paradigm; our SMCC-DT framework complements them by addressing the \emph{physical-layer} infrastructure needed to keep any DT---including security-oriented ones---synchronized without dedicated sensor networks.

\subsection{Integrated Sensing and Communication (ISAC)}
The ISAC paradigm has attracted substantial research interest as a key enabler of 6G networks~\cite{Liu2022ISACSurvey,Zhang2021ISAC6G,Cui2021ISAC}.
Liu \emph{et al.}~\cite{Liu2022ISACSurvey} provided an extensive survey of dual-function radar-communication (DFRC) waveform design, covering both shared and separated antenna architectures.
Zhang \emph{et al.}~\cite{Zhang2021ISAC6G} positioned ISAC within the broader 6G roadmap, arguing that perceptive networks will become a native feature of future cellular systems.
At the signal processing level, Chiriyath \emph{et al.}~\cite{Chiriyath2017ISAC} derived the fundamental performance bounds relating radar estimation accuracy (via the Cram\'{e}r--Rao Lower Bound) to communication capacity (via the Shannon limit) under shared spectrum allocation.
Liu \emph{et al.}~\cite{Liu2020ISACBeam} extended this framework to multiuser MIMO configurations with joint transmit beamforming.
Kumari \emph{et al.}~\cite{Kumari2019ISAC5G} demonstrated the feasibility of joint vehicular communication-radar using IEEE 802.11ad waveforms, while Dai \emph{et al.}~\cite{Dai2020ISAC} explored hybrid precoding for millimeter-wave massive MIMO with simultaneous information and power transfer.
More recently, Chen \emph{et al.}~\cite{Chen2023ISAC} proposed a joint beamforming design for ISAC systems that balances sensing and communication quality of service, and Liu \emph{et al.}~\cite{Liu2023ISACDRL} applied DRL to resource allocation for joint radar-communication in vehicular networks.
While these works provide the physical-layer foundations, none of them considers the downstream memory and computation constraints imposed by AI-driven DT workloads, which is the gap our SMCC-DT framework fills.

\subsection{Edge Intelligence and SMCC}
The convergence of AI and Edge Computing has given rise to Edge Intelligence, wherein deep learning models are deployed on resource-constrained Edge servers for real-time inference~\cite{Shi2016EdgeComputing,Mao2017MEC,Wang2020EdgeAI}.
The vision of 6G networks as AI-native platforms has been articulated by Saad \emph{et al.}~\cite{Saad2020_6G}, Tariq \emph{et al.}~\cite{Tariq2020_6G}, and Zhang \emph{et al.}~\cite{Zhang2019_6G}, all emphasizing the tight integration of sensing, communication, and computing.
Letaief \emph{et al.}~\cite{Letaief2024SMCC} recently articulated the SMCC vision, arguing that sensing, memory, communication, and computation must be co-designed rather than treated as independent layers.
Zhu \emph{et al.}~\cite{Zhu2024SMCC} formalized this perspective with a cross-layer resource allocation framework for SMCC pipelines, demonstrating significant efficiency gains over layered approaches.
Xu \emph{et al.}~\cite{Xu2023SMCCDT} further extended task-oriented sensing--communication--computing integration to multi-device edge AI scenarios.
Lyu \emph{et al.}~\cite{Lyu2023ISACDT} proposed an integrated sensing, communication, and computation framework specifically for over-the-air Digital Twin updates, while Zheng \emph{et al.}~\cite{Zheng2022DTCPS} addressed communication-efficient DT synchronization through semantic compression.
Wang \emph{et al.}~\cite{Wang2022DTResource} and Xu \emph{et al.}~\cite{Xu2021DTMEC} investigated DT-assisted resource allocation for IoT systems and energy-harvesting mobile edge computing, respectively.
Chen \emph{et al.}~\cite{Chen2021FedEdge} proposed a joint learning and communications framework for federated learning over wireless networks, while Wen \emph{et al.}~\cite{Wen2023SemCom} surveyed semantic communication as a bandwidth-efficient paradigm for AI-driven applications.
In the DRL domain, Mnih \emph{et al.}~\cite{Mnih2015DQN} pioneered deep Q-networks for discrete action spaces, while Schulman \emph{et al.}~\cite{Schulman2017PPO} introduced PPO for continuous control with stable policy updates.
Luong \emph{et al.}~\cite{Luong2019DRLNetwork} surveyed DRL applications in networking and resource management, establishing DRL as a viable optimizer for NP-hard wireless resource allocation problems.
He \emph{et al.}~\cite{He2017DRLResource} and Xiong \emph{et al.}~\cite{Xiong2019DRL5G} demonstrated DRL-based resource management for software-defined networks and 5G/beyond systems.
Almasan \emph{et al.}~\cite{Almasan2022DRLNetwork} combined DRL with graph neural networks for routing optimization, and Zappone \emph{et al.}~\cite{Zappone2019DVFS} analyzed the complementary roles of model-based and AI-based approaches in wireless network design.
Yang \emph{et al.}~\cite{Yang2020IntelligentRIS} investigated energy-efficient wireless communications using reconfigurable intelligent surfaces (RIS), a technology that our future work aims to integrate into the SMCC-DT pipeline.

\begin{figure*}[!t]
\centering
\begin{tikzpicture}[
  >=Stealth,
  block/.style={
    rectangle, draw, rounded corners=4pt,
    minimum width=1.9cm, minimum height=1.1cm,   
    font=\footnotesize\sffamily, align=center, thick
  },
  bigblock/.style={
    rectangle, draw, rounded corners=6pt,
    minimum width=3.0cm, minimum height=1.4cm,
    font=\footnotesize\sffamily, align=center,
    very thick, line width=1.2pt
  },
  arrow/.style={->, thick, >=Stealth},
  dasharrow/.style={->, thick, dashed, >=Stealth},
  lbl/.style={font=\scriptsize\sffamily, align=center},
  phase/.style={font=\scriptsize\sffamily\bfseries, rounded corners=2pt,
    inner sep=3pt, minimum height=0.5cm}
]

\node[phase, fill=sensing!20,  text=sensing] at (0,   3.0) {SENSING};
\node[phase, fill=commu!20,    text=commu]   at (5.2, 3.0) {COMMUNICATION};
\node[phase, fill=memory!20,   text=memory]  at (10.2,3.0) {MEMORY};
\node[phase, fill=compu!20,    text=compu]   at (14.4,3.0) {COMPUTATION};

\node[bigblock, fill=lightbg, text=edgecol]          (env) at (0,    1.8) {Physical\\Environment\\($\mathcal{K}$ assets)};
\node[bigblock, fill=sensing!12, draw=sensing]        (bs)  at (0,   -0.2) {6G Base Station\\ISAC Transmitter};

\node[block, fill=sensing!15, draw=sensing]  (sense) at ( 3.0, -0.2) {Radar\\Processing\\$\hat{\boldsymbol{\theta}}_k$};
\node[block, fill=commu!15,   draw=commu]    (chan)  at ( 5.5, -0.2) {Wireless\\Channel\\$\mathbf{H}$};
\node[block, fill=commu!15,   draw=commu]    (rx)   at ( 7.8, -0.2) {Edge\\Receiver\\$\hat{\mathbf{s}}$};

\node[block, fill=memory!15, draw=memory] (mem1) at (10.2,  1.5) {Model\\Weights\\$M_{\text{model}}$};
\node[block, fill=memory!15, draw=memory] (mem2) at (10.2, -0.2) {Data\\Buffer\\$M_{\text{buf}}$};

\node[block,    fill=compu!15, draw=compu]             (cpu) at (13.0, 0.8) {AI Inference\\Engine\\$f_{\text{cpu}}$};
\node[bigblock, fill=dtcol!20, draw=dtcol, line width=1.5pt] (dt) at (16.3, 0.8) {Digital Twin\\State $\mathbf{s}_k^{(t)}$};

\draw[arrow, sensing] (bs.north) -- node[left, lbl, text=sensing]{ISAC\\signal} (env.south);

\draw[arrow, sensing,
      decorate, decoration={snake, amplitude=1.2mm, segment length=4mm}]
  (env.east) -- node[above right, lbl, text=sensing, pos=0.5]{Echoes} (sense.north);

\draw[arrow, commu] (bs.east)    -- node[below, lbl, text=commu]{Data} (sense.west);
\draw[arrow, commu] (sense.east) -- (chan.west);
\draw[arrow, commu] (chan.east)  -- (rx.west);

\draw[arrow, memory] (rx.east)  -- (mem2.west);          
\draw[arrow, memory] (rx.north) |- (mem1.west);          

\draw[arrow, compu] (mem1.east) -- (cpu.west);
\draw[arrow, compu] (mem2.east) -| (cpu.south);

\draw[arrow, dtcol, line width=1.5pt]
  (cpu.east) -- node[above, lbl, text=dtcol]{Update} (dt.west);

\draw[dasharrow, edgecol]
  (dt.south) -- ++(0,-1.2) -|
  node[below, lbl, pos=0.25, text=edgecol]{Policy feedback $\pi_\phi$}
  (bs.south);

\begin{scope}[on background layer]
  \node[draw=memory, dashed, thick, rounded corners=6pt,
    fit=(mem1)(mem2), inner sep=6pt,
    label={[font=\scriptsize\sffamily, text=memory]above:Edge Memory $M_{\max}$}] {};
\end{scope}

\end{tikzpicture}
\caption{The SMCC-DT architecture. A 6G base station transmits ISAC waveforms that sense
$K$ physical assets and communicate the extracted telemetry to the Edge server.
The Edge memory is partitioned between AI model weights and the data buffer, while the
CPU executes the inference engine that updates the Digital Twin state.
A DRL policy $\pi_\phi$ feeds back resource allocation decisions to the ISAC transmitter.}
\label{fig:architecture}
\end{figure*}
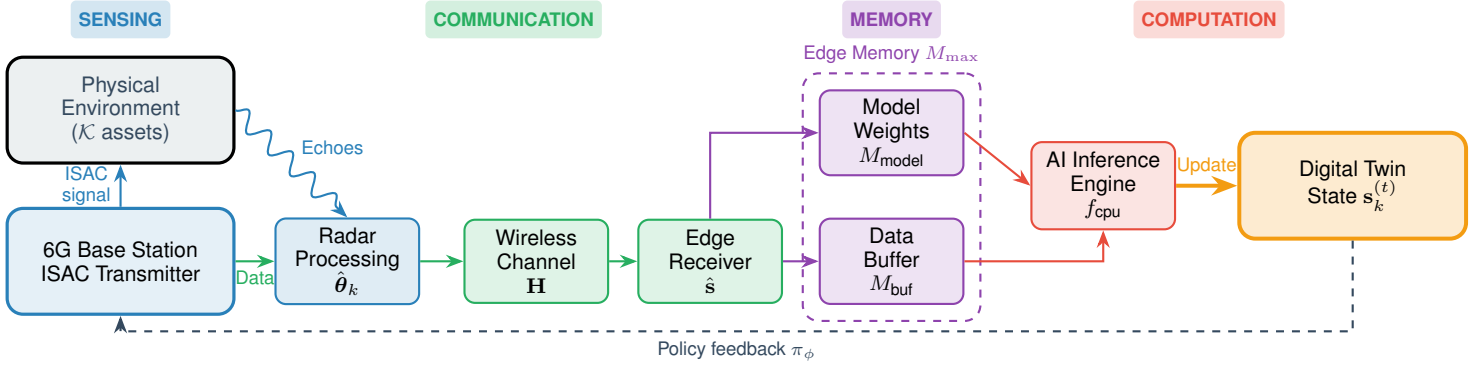

\section{System Model}
\label{sec:system}

We consider an industrial IoT environment comprising a set $\mathcal{K} = \{1, 2, \ldots, K\}$ of physical assets (e.g., robotic arms, conveyor segments, autonomous vehicles) monitored by a DT hosted on an Edge server.
A 6G base station (BS), co-located with the Edge server, transmits ISAC waveforms that simultaneously sense the physical environment and communicate data payloads to IoT devices.
Fig.~\ref{fig:architecture} illustrates the SMCC-DT architecture.

\subsection{Sensing Model}
\label{subsec:sensing}

The BS transmits a dual-function ISAC waveform $\mathbf{x}(t) \in \mathbb{C}^{N_t \times 1}$ through $N_t$ antennas with beamforming matrix $\mathbf{W} = [\mathbf{w}_s, \mathbf{w}_c] \in \mathbb{C}^{N_t \times 2}$, where $\mathbf{w}_s$ and $\mathbf{w}_c$ denote the sensing and communication beamforming vectors, respectively.
The total transmit power is bounded by $\|\mathbf{w}_s\|^2 + \|\mathbf{w}_c\|^2 \leq P_{\max}$.

For the $k$-th physical asset located at position $\mathbf{p}_k \in \real^3$, the reflected echo received by the BS is
\begin{equation}
\mathbf{y}_k^{(s)}(t) = \alpha_k \, \mathbf{a}(\theta_k, \phi_k) \, \mathbf{a}^H(\theta_k, \phi_k) \, \mathbf{w}_s \, x(t - \tau_k) + \mathbf{n}_s(t),
\label{eq:echo}
\end{equation}
where $\alpha_k$ is the complex radar cross-section coefficient, $\mathbf{a}(\theta_k, \phi_k)$ is the steering vector at azimuth $\theta_k$ and elevation $\phi_k$, $\tau_k = 2\|\mathbf{p}_k\|/c$ is the round-trip delay, and $\mathbf{n}_s(t) \sim \mathcal{CN}(\mathbf{0}, \sigma_s^2 \mathbf{I})$ is additive white Gaussian noise.

The sensing accuracy for estimating the parameter vector $\boldsymbol{\theta}_k = [\theta_k, \phi_k, \tau_k, f_{d,k}]^T$ (including Doppler shift $f_{d,k}$) is bounded by the Cram\'{e}r--Rao Lower Bound (CRLB):
\begin{equation}
\text{Var}(\hat{\boldsymbol{\theta}}_k) \geq \mathbf{J}^{-1}(\boldsymbol{\theta}_k),
\label{eq:crlb}
\end{equation}
where $\mathbf{J}(\boldsymbol{\theta}_k)$ is the Fisher Information Matrix (FIM).
For the monostatic ISAC configuration, the FIM depends on the sensing beamforming power $\rho_s = \|\mathbf{w}_s\|^2$ through~\cite{Chiriyath2017ISAC}:
\begin{equation}
\text{tr}\bigl(\mathbf{J}^{-1}(\boldsymbol{\theta}_k)\bigr) = \frac{\sigma_s^2}{|\alpha_k|^2 \, \rho_s \, T_s \, B_s} \triangleq \epsilon_k(\rho_s),
\label{eq:sensing_acc}
\end{equation}
where $T_s$ is the sensing integration time and $B_s$ is the sensing bandwidth.
We define the sensing quality requirement as $\epsilon_k(\rho_s) \leq \epsilon_{\max}$, ensuring a minimum estimation accuracy.

\subsection{Communication Model}
\label{subsec:communication}

Simultaneously, the communication component of the ISAC waveform carries the sensed telemetry $\hat{\boldsymbol{\theta}}_k$ from the BS to the Edge server over a channel $\mathbf{H} \in \mathbb{C}^{N_r \times N_t}$.
The received signal at the Edge is
\begin{equation}
\mathbf{y}^{(c)}(t) = \mathbf{H} \, \mathbf{w}_c \, s_c(t) + \underbrace{\mathbf{H} \, \mathbf{w}_s \, x(t)}_{\text{sensing interference}} + \mathbf{n}_c(t),
\label{eq:comm_rx}
\end{equation}
where $s_c(t)$ is the communication symbol and $\mathbf{n}_c(t) \sim \mathcal{CN}(\mathbf{0}, \sigma_c^2 \mathbf{I})$.
After interference cancellation, the achievable data rate is
\begin{equation}
R = B_c \log_2 \!\left(1 + \frac{\|\mathbf{H} \mathbf{w}_c\|^2}{\sigma_c^2 + \gamma \|\mathbf{H} \mathbf{w}_s\|^2}\right),
\label{eq:rate}
\end{equation}
where $B_c$ is the communication bandwidth and $\gamma \in [0,1]$ is the residual interference factor after sensing signal cancellation.
The communication constraint requires $R \geq R_{\min}$, ensuring sufficient throughput for DT telemetry delivery.

\subsection{Memory Model}
\label{subsec:memory}

The Edge server has a total memory capacity $M_{\max}$ (in bytes), which must be partitioned between two competing demands:

\begin{enumerate}
\item \emph{AI Model Storage} ($M_{\text{model}}$): The DT inference engine requires loading the weights of a large-scale AI model (e.g., a vision transformer or a physics-informed neural network). For a model with $L$ layers and parameter count $\Omega$, the memory footprint is
\begin{equation}
M_{\text{model}} = \beta \cdot \Omega + \sum_{\ell=1}^{L} A_\ell,
\label{eq:mem_model}
\end{equation}
where $\beta$ is the bytes-per-parameter (e.g., $\beta = 2$ for FP16 quantization) and $A_\ell$ is the activation memory for layer $\ell$.

\item \emph{Data Buffer} ($M_{\text{buf}}$): Incoming telemetry data from the ISAC pipeline must be buffered before inference.
For $K$ assets sampled at rate $f_s$ with per-sample size $d_k$, the buffer requirement over one DT update cycle of duration $T_{\text{cycle}}$ is
\begin{equation}
M_{\text{buf}} = \sum_{k=1}^{K} d_k \cdot f_s \cdot T_{\text{cycle}}.
\label{eq:mem_buf}
\end{equation}
\end{enumerate}

The memory constraint is $M_{\text{model}} + M_{\text{buf}} \leq M_{\max}$.
A trade-off exists: allocating more memory to the model (larger $\Omega$) improves inference accuracy but reduces the buffer capacity, potentially causing data drops.

\subsection{Computation Model}
\label{subsec:computation}

The Edge CPU operates at a tunable frequency $f_{\text{cpu}} \in [f_{\min}, f_{\max}]$ (in cycles per second).
Processing the buffered telemetry through the AI model of size $\Omega$ requires
\begin{equation}
C_{\text{inf}} = \kappa \cdot \Omega \cdot K
\label{eq:comp_cycles}
\end{equation}
CPU cycles, where $\kappa$ is the cycles-per-parameter-per-asset factor that depends on the model architecture.
The inference latency is
\begin{equation}
T_{\text{inf}} = \frac{C_{\text{inf}}}{f_{\text{cpu}}} = \frac{\kappa \, \Omega \, K}{f_{\text{cpu}}},
\label{eq:inf_latency}
\end{equation}
and the associated energy consumption follows the cubic DVFS model~\cite{Mao2017MEC}:
\begin{equation}
E_{\text{inf}} = \xi \, C_{\text{inf}} \, f_{\text{cpu}}^2,
\label{eq:inf_energy}
\end{equation}
where $\xi$ is the effective capacitance coefficient of the processor.

\begin{figure}[!t]
\centering
\begin{tikzpicture}[
  >=Stealth,
  block/.style={rectangle, draw, rounded corners=3pt, minimum width=1.4cm,
    minimum height=0.6cm, font=\scriptsize\sffamily, align=center, thick},
  arrow/.style={->, thick},
  lbl/.style={font=\tiny\sffamily, align=center}
]

\node[block, fill=sensing!15, draw=sensing] (tx) at (0,0) {ISAC\\Tx};

\node[block, fill=sensing!10, draw=sensing] (bfs) at (0,  1.3) {$\mathbf{w}_s$\\Sensing};
\node[block, fill=commu!10,   draw=commu]   (bfc) at (0, -1.3) {$\mathbf{w}_c$\\Comms};

\draw[arrow, sensing] (tx.north) -- (bfs.south);
\draw[arrow, commu]   (tx.south) -- (bfc.north);

\node[draw=edgecol, thick, rounded corners=3pt, fill=lightbg,
  font=\tiny\sffamily, align=center, inner sep=3pt] (pwr) at (0, -2.3)
  {$\|\mathbf{w}_s\|^2 + \|\mathbf{w}_c\|^2 \leq P_{\max}$};

\node[block, fill=dtcol!15, draw=dtcol, minimum width=1.8cm]
      (tgt) at (2.8, 1.3) {Target $k$\\$\mathbf{p}_k, \alpha_k$};

\draw[arrow, sensing] (bfs.east) -- node[above, lbl]{$\mathbf{x}(t)$} (tgt.west);
\coordinate (refl) at (2.8, 0.6);
\draw[arrow, sensing, decorate, decoration={snake, amplitude=0.6mm, segment length=2.5mm}]
  (tgt.south) -- (refl) -- (1.1, 0.6) node[above, lbl]{$\mathbf{y}_k^{(s)}$};

\node[block, fill=sensing!20, draw=sensing] (sp) at (2.8, -0.1) {FIM\\$\mathbf{J}(\boldsymbol{\theta}_k)$};
\draw[arrow, sensing] (refl) -- (sp.north);

\node[font=\scriptsize\sffamily, text=sensing] (crlb) at (5.0, -0.1)
      {$\epsilon_k \leq \epsilon_{\max}$};
\draw[arrow, sensing] (sp.east) -- (crlb.west);

\node[block, fill=commu!15, draw=commu] (ch) at (2.8, -1.3) {Channel\\$\mathbf{H}$};
\draw[arrow, commu] (bfc.east) -- node[above, lbl]{$s_c(t)$} (ch.west);

\node[block, fill=commu!20, draw=commu] (erx) at (5.0, -1.3) {Edge Rx};
\draw[arrow, commu] (ch.east) -- node[above, lbl]{$\mathbf{y}^{(c)}$} (erx.west);

\node[font=\scriptsize\sffamily, text=commu] (rate) at (5.0, -2.1)
      {$R \geq R_{\min}$};
\draw[arrow, commu] (erx.south) -- (rate.north);

\begin{scope}[on background layer]
  \node[draw=gridgray, dashed, rounded corners=4pt, fit=(tx)(bfs)(bfc)(pwr), inner sep=4pt] (txbox) {};
  \node[font=\tiny\sffamily, text=gray, above right] at (txbox.south west) {Base Station};
\end{scope}

\end{tikzpicture}
\caption{ISAC signal model. The dual-function waveform splits into sensing ($\mathbf{w}_s$) and communication ($\mathbf{w}_c$) paths. Both share the power budget $P_{\max}$, while satisfying CRLB and rate constraints.}
\label{fig:isac_model}
\end{figure}
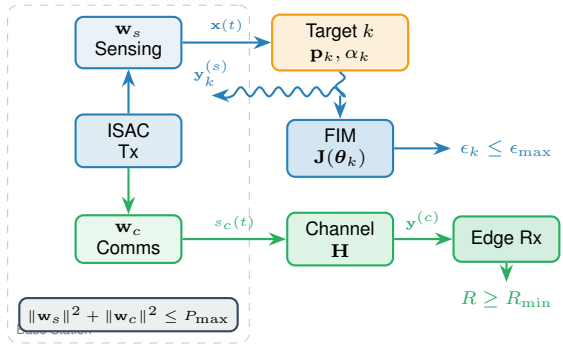

\subsection{DT Synchronization Latency}
\label{subsec:sync}

The end-to-end DT synchronization latency for one update cycle comprises four sequential stages aligned with the SMCC dimensions:
\begin{equation}
\Tsync = \underbrace{T_s}_{\text{Sensing}} + \underbrace{\frac{D_{\text{tel}}}{R}}_{\text{Communication}} + \underbrace{T_{\text{load}}(\Omega)}_{\text{Memory}} + \underbrace{\frac{\kappa \Omega K}{f_{\text{cpu}}}}_{\text{Computation}},
\label{eq:tsync}
\end{equation}
where $D_{\text{tel}} = \sum_k |\hat{\boldsymbol{\theta}}_k| \cdot b$ is the total telemetry payload (bits), and $T_{\text{load}}(\Omega) = M_{\text{model}} / B_{\text{mem}}$ is the model loading time from storage to active memory at bandwidth $B_{\text{mem}}$.

\section{Cross-Layer Optimization Problem}
\label{sec:problem}

We seek to minimize the DT synchronization latency $\Tsync$ by jointly optimizing the sensing power $\rho_s$, the communication beamformer $\mathbf{w}_c$, the memory partition $(M_{\text{model}}, M_{\text{buf}})$, and the CPU frequency $f_{\text{cpu}}$.
Let $\boldsymbol{x} = (\rho_s, \mathbf{w}_c, M_{\text{model}}, M_{\text{buf}}, f_{\text{cpu}})$ denote the decision vector.

\begin{figure}[!t]
\centering
\begin{tikzpicture}[
  >=Stealth,
  scale=1, every node/.style={transform shape},
  layer/.style={
    rectangle, draw, rounded corners=5pt,
    minimum width=3.3cm, minimum height=1.0cm,
    font=\footnotesize\sffamily, align=center, thick
  },
  coupling/.style={<->, thick, densely dashed},
  clabel/.style={font=\tiny\sffamily, fill=white, inner sep=1.5pt}
]

\node[layer, fill=sensing!15, draw=sensing] (S) at (0, 4.5) {Sensing\\$\rho_s, T_s, B_s$};
\node[layer, fill=commu!15, draw=commu] (C1) at (0, 3.0) {Communication\\$\mathbf{w}_c, R, B_c$};
\node[layer, fill=memory!15, draw=memory] (M) at (0, 1.5) {Memory\\$M_{\text{model}}, M_{\text{buf}}$};
\node[layer, fill=compu!15, draw=compu] (C2) at (0, 0.0) {Computation\\$f_{\text{cpu}}, \Omega, \kappa$};

\draw[coupling, sensing!70!black] (S.east) -- ++(1.2,0) |- node[clabel, pos=0.25] {Power} (C1.east);
\draw[coupling, commu!70!black] (C1.east) -- ++(1.8,0) |- node[clabel, pos=0.25] {Throughput $\to$ Buffer} (M.east);
\draw[coupling, memory!70!black] (M.east) -- ++(1.2,0) |- node[clabel, pos=0.25] {Model size $\to$ Cycles} (C2.east);
\draw[coupling, sensing!50!compu] (S.west) -- ++(-1.5,0) |- node[clabel, pos=0.25] {Accuracy $\to$ Model} (C2.west);

\node[draw=dtcol, fill=dtcol!10, very thick, rounded corners=4pt,
  font=\footnotesize\sffamily, align=center, minimum width=4cm] (obj) at (0, -1.5) 
  {$\min_{\boldsymbol{x}} \; \Tsync(\boldsymbol{x})$};
\draw[->, very thick, dtcol] (C2.south) -- (obj.north);

\end{tikzpicture}
\caption{Cross-layer coupling in the SMCC optimization. Each pair of adjacent layers shares resources or imposes constraints on the other. The dashed arrows indicate the key trade-off channels: power allocation between Sensing and Communication, throughput-to-buffer coupling between Communication and Memory, model size coupling between Memory and Computation, and the end-to-end accuracy-to-model dependency between Sensing and Computation.}
\label{fig:coupling}
\end{figure}
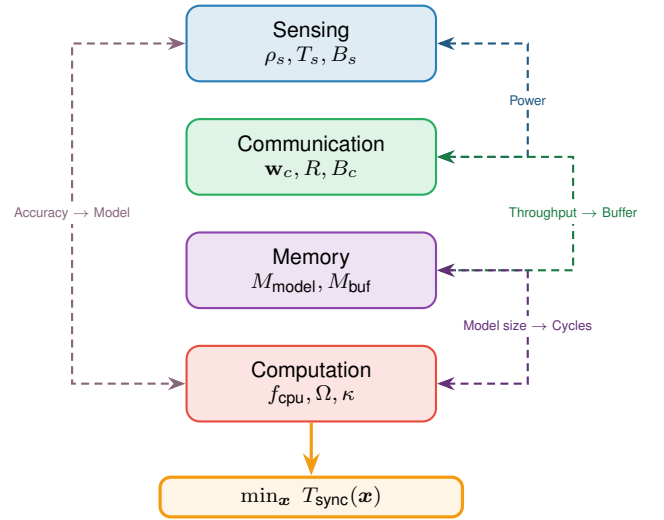

The optimization problem is formulated as:
\begin{subequations}
\label{eq:P1}
\begin{align}
(\mathcal{P}_1): \quad & \min_{\boldsymbol{x}} \quad \Tsync(\boldsymbol{x}) \label{eq:P1obj}\\
\text{s.t.} \quad
& \epsilon_k(\rho_s) \leq \epsilon_{\max}, \quad \forall k \in \mathcal{K}, \label{eq:P1c1}\\
& R(\mathbf{w}_c, \rho_s) \geq R_{\min}, \label{eq:P1c2}\\
& \rho_s + \|\mathbf{w}_c\|^2 \leq P_{\max}, \label{eq:P1c3}\\
& M_{\text{model}} + M_{\text{buf}} \leq M_{\max}, \label{eq:P1c4}\\
& M_{\text{model}} \geq \beta \, \Omega_{\min}, \label{eq:P1c5}\\
& f_{\min} \leq f_{\text{cpu}} \leq f_{\max}, \label{eq:P1c6}\\
& E_{\text{inf}}(f_{\text{cpu}}, \Omega) \leq E_{\max}, \label{eq:P1c7}\\
& \Omega \in \mathcal{Z}^+, \label{eq:P1c8}
\end{align}
\end{subequations}
where~\eqref{eq:P1c1} ensures sensing accuracy,~\eqref{eq:P1c2} guarantees minimum communication throughput,~\eqref{eq:P1c3} enforces the power budget,~\eqref{eq:P1c4} respects the memory capacity,~\eqref{eq:P1c5} requires a minimum model size for acceptable inference quality,~\eqref{eq:P1c6} bounds the CPU frequency,~\eqref{eq:P1c7} limits the energy budget, and~\eqref{eq:P1c8} restricts the model parameter count to integers (corresponding to discrete model variants).

\begin{theorem}[NP-Hardness]
\label{thm:nphard}
Problem $\mathcal{P}_1$ is NP-hard.
\end{theorem}
\begin{IEEEproof}
We prove NP-hardness by reduction from the 0-1 Knapsack problem.
Consider a simplified instance of $\mathcal{P}_1$ with a single asset ($K=1$), fixed sensing power $\rho_s$, and fixed CPU frequency $f_{\text{cpu}}$.
The remaining decision is to select a model variant $\Omega \in \{\Omega_1, \ldots, \Omega_N\}$ and allocate the memory partition $(M_{\text{model}}, M_{\text{buf}})$ subject to the capacity constraint $M_{\text{model}} + M_{\text{buf}} \leq M_{\max}$.
Each model variant $\Omega_i$ has a distinct inference quality $q_i$ and memory footprint $m_i$.
Selecting the model that minimizes latency while satisfying both memory and quality constraints reduces to a knapsack instance, which is known to be NP-hard~\cite{Garey1979NPHard}.
Since this restricted case is NP-hard, the general problem $\mathcal{P}_1$ is NP-hard \emph{a fortiori}.
\end{IEEEproof}

\section{DRL-Based SMCC Resource Allocation}
\label{sec:algorithm}

Given the NP-hardness of $\mathcal{P}_1$, we design a DRL agent, termed \textsc{SmccAgent}, that learns a policy $\pi_\phi(\mathbf{a}_t | \mathbf{o}_t)$ mapping real-time system observations to resource allocation actions.
We adopt the Proximal Policy Optimization (PPO) algorithm~\cite{Schulman2017PPO} for its stability under continuous action spaces and compatibility with constraint handling via reward shaping.

\subsection{Markov Decision Process Formulation}

We model the SMCC resource allocation as a Markov Decision Process (MDP) $(\mathcal{S}, \mathcal{A}, P, r, \gamma_d)$:

\emph{State space} $\mathcal{S}$: At time step $t$, the agent observes
\begin{equation}
\mathbf{o}_t = \bigl[\mathbf{h}_t,\, \boldsymbol{\epsilon}_t,\, m_t^{\text{free}},\, f_t^{\text{load}},\, q_t^{\text{buf}}\bigr],
\label{eq:state}
\end{equation}
where $\mathbf{h}_t \in \mathbb{R}^{N_r N_t}$ is the vectorized channel state, $\boldsymbol{\epsilon}_t = [\epsilon_1, \ldots, \epsilon_K]^T$ is the current sensing error vector, $m_t^{\text{free}}$ is the available memory, $f_t^{\text{load}}$ is the CPU load factor, and $q_t^{\text{buf}}$ is the buffer occupancy.

\emph{Action space} $\mathcal{A}$: The agent outputs a hierarchical action
\begin{equation}
\mathbf{a}_t = \bigl[\rho_s^{(t)},\, \mathbf{w}_c^{(t)},\, \Delta M^{(t)},\, f_{\text{cpu}}^{(t)},\, \omega^{(t)}\bigr],
\label{eq:action}
\end{equation}
where $\Delta M^{(t)}$ is the memory reallocation increment and $\omega^{(t)} \in \{1, \ldots, N\}$ is the model variant selector (discrete).

\emph{Reward}: We design a composite reward that balances latency minimization with constraint satisfaction:
\begin{equation}
r_t = -\lambda_1 \Tsync^{(t)} - \lambda_2 \, \Etot^{(t)} - \lambda_3 \sum_{i} \max(0, g_i(\boldsymbol{x}_t)),
\label{eq:reward}
\end{equation}
where $g_i(\boldsymbol{x}_t)$ are the constraint violation magnitudes from~\eqref{eq:P1c1}--\eqref{eq:P1c7} and $\lambda_1, \lambda_2, \lambda_3 > 0$ are weighting coefficients.

\subsection{Network Architecture}

The \textsc{SmccAgent} employs a dual-head architecture with shared feature extraction:

\begin{itemize}
\item A shared encoder comprising three fully connected layers (256--128--64 neurons) with ReLU activations that maps $\mathbf{o}_t$ to a latent representation $\mathbf{z}_t$.
\item An \emph{actor head} that produces the continuous actions $(\rho_s, \mathbf{w}_c, \Delta M, f_{\text{cpu}})$ via a Gaussian policy and the discrete action $\omega$ via a categorical distribution.
\item A \emph{critic head} that estimates the state-value function $V_\psi(\mathbf{o}_t)$.
\end{itemize}

\subsection{Training Procedure}

Algorithm~\ref{alg:ppo} summarizes the \textsc{SmccAgent} training loop.
The agent interacts with the SMCC-DT environment for $T$ episodes, collecting trajectories $\{(\mathbf{o}_t, \mathbf{a}_t, r_t)\}$ and updating the policy using the PPO clipped surrogate objective:
\begin{equation}
L^{\text{PPO}}(\phi) = \mathbb{E}_t \!\left[\min\!\left(\frac{\pi_\phi(\mathbf{a}_t|\mathbf{o}_t)}{\pi_{\phi_{\text{old}}}(\mathbf{a}_t|\mathbf{o}_t)} \hat{A}_t,\; \text{clip}(\cdot, 1\!-\!\varepsilon, 1\!+\!\varepsilon) \hat{A}_t\right)\right],
\label{eq:ppo_loss}
\end{equation}
where $\hat{A}_t$ is the Generalized Advantage Estimate (GAE) and $\varepsilon = 0.2$ is the clipping parameter.

\begin{algorithm}[!t]
\caption{\textsc{SmccAgent}: PPO-Based SMCC Resource Allocation}
\label{alg:ppo}
\begin{algorithmic}[1]
\REQUIRE Environment $\mathcal{E}$, policy $\pi_\phi$, critic $V_\psi$, episodes $T$
\STATE Initialize $\phi$, $\psi$ randomly
\FOR{episode $= 1, 2, \ldots, T$}
  \STATE Reset environment; observe $\mathbf{o}_0$
  \FOR{$t = 0, 1, \ldots, T_{\max}$}
    \STATE Sample $\mathbf{a}_t \sim \pi_\phi(\cdot | \mathbf{o}_t)$
    \STATE Execute $\mathbf{a}_t$; observe $\mathbf{o}_{t+1}$, $r_t$
    \STATE Store $(\mathbf{o}_t, \mathbf{a}_t, r_t, \mathbf{o}_{t+1})$ in buffer $\mathcal{B}$
  \ENDFOR
  \STATE Compute GAE advantages $\hat{A}_t$ from $\mathcal{B}$
  \FOR{$e = 1, \ldots, E_{\text{epochs}}$}
    \STATE Update $\phi$ via~\eqref{eq:ppo_loss} using mini-batches from $\mathcal{B}$
    \STATE Update $\psi$ via MSE loss: $L^V = \frac{1}{|\mathcal{B}|}\sum_t (V_\psi(\mathbf{o}_t) - \hat{R}_t)^2$
  \ENDFOR
\ENDFOR
\RETURN Trained policy $\pi_\phi$
\end{algorithmic}
\end{algorithm}

\begin{figure}[!t]
\centering
\begin{tikzpicture}[
  >=Stealth,
  scale=0.78, every node/.style={transform shape},
  block/.style={rectangle, draw, rounded corners=3pt, minimum width=1.4cm,
    minimum height=0.7cm, font=\scriptsize\sffamily, align=center, thick},
  wide/.style={rectangle, draw, rounded corners=3pt, minimum width=2.6cm,
    minimum height=0.7cm, font=\scriptsize\sffamily, align=center, thick},
  arrow/.style={->, thick}
]

\node[wide, fill=lightbg, draw=edgecol] (obs) at (0, -0.2) {Observation $\mathbf{o}_t$};

\node[wide, fill=edgecol!15, draw=edgecol] (fc1) at (0, 1.1) {FC 256, ReLU};
\node[wide, fill=edgecol!15, draw=edgecol] (fc2) at (0, 2.0) {FC 128, ReLU};
\node[wide, fill=edgecol!15, draw=edgecol] (fc3) at (0, 2.9) {FC 64,  ReLU};

\draw[arrow] (obs) -- (fc1);
\draw[arrow] (fc1) -- (fc2);
\draw[arrow] (fc2) -- (fc3);

\node[font=\scriptsize\sffamily, text=edgecol] at (1.75, 2.9) {$\mathbf{z}_t$};

\node[block, fill=commu!15,  draw=commu]  (act_c) at (-2.2, 4.1) {Gaussian\\$\mu,\sigma$};
\node[block, fill=sensing!15, draw=sensing](act_d) at ( 0.0, 4.1) {Categorical\\$p(\omega)$};
\node[block, fill=compu!15,  draw=compu]  (crit)  at ( 2.2, 4.1) {Value\\$V_\psi$};

\draw[arrow, commu]   (fc3.north) -- (act_c.south);
\draw[arrow, sensing] (fc3.north) -- (act_d.south);
\draw[arrow, compu]   (fc3.north) -- (crit.south);

\node[block, fill=commu!8,  draw=commu,  minimum width=1.9cm]
      (out_c) at (-2.2, 5.3) {$\rho_s,\mathbf{w}_c,\Delta M,f_{\text{cpu}}$};
\node[block, fill=sensing!8, draw=sensing]
      (out_d) at ( 0.0, 5.3) {$\omega^{(t)}$};
\node[block, fill=compu!8,  draw=compu]
      (out_v) at ( 2.2, 5.3) {$\hat{V}(\mathbf{o}_t)$};

\draw[arrow, commu]   (act_c) -- (out_c);
\draw[arrow, sensing] (act_d) -- (out_d);
\draw[arrow, compu]   (crit)  -- (out_v);

\node[font=\scriptsize\sffamily\bfseries, text=commu]   at (-2.2, 6.0) {Actor (Cont.)};
\node[font=\scriptsize\sffamily\bfseries, text=sensing]  at ( 0.0, 6.0) {Actor (Disc.)};
\node[font=\scriptsize\sffamily\bfseries, text=compu]   at ( 2.2, 6.0) {Critic};

\begin{scope}[on background layer]
\node[draw=edgecol, dashed, thick, rounded corners=5pt,
  fit=(fc1)(fc2)(fc3), inner xsep=8pt, inner ysep=6pt,
  label={[font=\scriptsize\sffamily\bfseries, text=edgecol]right:Shared Encoder}] {};
\end{scope}

\end{tikzpicture}
\caption{The \textsc{SmccAgent} dual-head neural network architecture.
A shared encoder maps observations to a latent representation $\mathbf{z}_t$,
which feeds two actor heads (Gaussian for continuous actions, Categorical for
model selection) and a critic head for value estimation.}
\label{fig:drl_arch}
\end{figure}
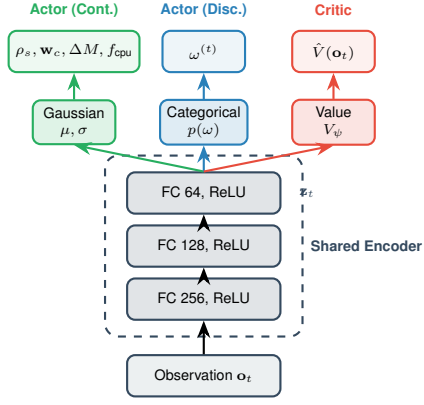

\section{Performance Evaluation}
\label{sec:evaluation}

\subsection{Simulation Setup}

We evaluate SMCC-DT through extensive Monte Carlo simulations on a 500-node industrial IoT testbed.
Table~\ref{tab:params} summarizes the key simulation parameters.
The physical environment models a $200 \times 200 \times 50$~m$^3$ smart factory with $K = 500$ assets distributed across three production zones.
The 6G BS is equipped with $N_t = 64$ transmit antennas operating at 28~GHz with $P_{\max} = 40$~dBm.
The ISAC waveform uses $B_s = B_c = 100$~MHz bandwidth.
The Edge server has $M_{\max} = 32$~GB of memory and a CPU with $f_{\max} = 4.0$~GHz.
We evaluate three AI model sizes: a small model ($\Omega = 7\text{M}$, 14~MB), a medium model ($\Omega = 125\text{M}$, 250~MB), and a large model ($\Omega = 1.3\text{B}$, 2.6~GB), corresponding to lightweight, mid-range, and foundation-class architectures.

\begin{table}[!t]
\centering
\caption{Simulation Parameters}
\label{tab:params}
\renewcommand{\arraystretch}{0.95}
\begin{tabular}{lcc}
\toprule
\textbf{Parameter} & \textbf{Symbol} & \textbf{Value} \\
\midrule
Number of assets & $K$ & 500 \\
BS antennas & $N_t$ & 64 \\
Carrier frequency & $f_c$ & 28 GHz \\
Max transmit power & $P_{\max}$ & 40 dBm \\
Bandwidth (S and C) & $B_s, B_c$ & 100 MHz \\
Noise power & $\sigma_s^2, \sigma_c^2$ & $-$90 dBm \\
Edge memory & $M_{\max}$ & 32 GB \\
CPU frequency range & $[f_{\min}, f_{\max}]$ & [1.0, 4.0] GHz \\
DVFS capacitance & $\xi$ & $10^{-28}$ \\
Sensing accuracy threshold & $\epsilon_{\max}$ & $10^{-3}$ \\
Min comm. rate & $R_{\min}$ & 500 Mbps \\
Energy budget & $E_{\max}$ & 10 J/cycle \\
PPO learning rate & $\alpha_{\text{lr}}$ & $3 \times 10^{-4}$ \\
PPO clip parameter & $\varepsilon$ & 0.2 \\
Discount factor & $\gamma_d$ & 0.99 \\
Training episodes & $T$ & 10,000 \\
\bottomrule
\end{tabular}
\end{table}

\subsection{Baselines}

We compare SMCC-DT against the following baseline schemes:

\emph{Orthogonal Allocation (OA)}: Sensing and communication operate on separate frequency bands ($B_s/2$ each), losing the spectral efficiency gains of ISAC. Memory and computation are optimized independently.

\emph{Compute-Only Optimization (CO)}: The system optimizes only $f_{\text{cpu}}$ and $\Omega$, using fixed equal power split for sensing/communication and a static memory partition.

\emph{Greedy Heuristic (GH)}: A rule-based policy that allocates resources greedily, prioritizing sensing accuracy first, then communication throughput, and finally computation.

\emph{Random Allocation (RA)}: Uniform random sampling of all decision variables within their feasible ranges, providing a lower-bound reference.

\subsection{Convergence Analysis}

Fig.~\ref{fig:results}(a) shows the learning curve of \textsc{SmccAgent} over 10,000 training episodes.
The average reward converges after approximately 3,500 episodes, with the policy stabilizing to a near-optimal allocation.
The constraint violation rate drops below 1\% after 2,000 episodes, confirming that the reward shaping mechanism effectively internalizes the SMCC constraints.

\subsection{Synchronization Latency}

Fig.~\ref{fig:results}(b) compares the DT synchronization latency across methods as a function of the number of monitored assets $K$.
SMCC-DT achieves the lowest latency across all asset counts.
At $K = 500$, SMCC-DT attains $\Tsync = 12.3$~ms, representing a 38.7\% reduction over OA ($\Tsync = 20.1$~ms) and a 27.2\% reduction over CO ($\Tsync = 16.9$~ms).
The improvement stems from the joint optimization of ISAC power splitting and memory-computation co-allocation, which eliminates the bottlenecks present in siloed approaches.

\subsection{Energy Efficiency}

Fig.~\ref{fig:results}(c) plots the total energy consumption versus sensing accuracy.
SMCC-DT achieves the most favorable Pareto frontier, maintaining sensing accuracy above 95\% while consuming 27.4\% less energy than OA and 19.8\% less than CO.
The energy savings arise primarily from intelligent CPU frequency scaling: \textsc{SmccAgent} learns to reduce $f_{\text{cpu}}$ when the model is smaller (less computation needed) and to compensate by allocating more memory to buffering, which smooths the inference workload.

\subsection{Impact of AI Model Size}

Fig.~\ref{fig:results}(d) examines the effect of the AI model memory footprint on the overall system performance.
As $\Omega$ increases from 7M to 1.3B parameters, the synchronization latency rises due to increased computation and memory pressure.
However, SMCC-DT degrades gracefully: the latency increase from medium to large model is only 31\% for SMCC-DT versus 67\% for CO, demonstrating the effectiveness of the joint memory-computation optimization.
Notably, OA fails to satisfy the latency constraint for the large model at $K > 300$ assets, while SMCC-DT remains feasible up to $K = 500$.

\begin{figure*}[!t]
\centering
\begin{tikzpicture}

\begin{axis}[
  name=plotA,
  at={(0,0)},
  width=4.4cm, height=3.8cm,
  xlabel={Episode ($\times 10^3$)},
  ylabel={Avg. Reward},
  xmin=0, xmax=10,
  ymin=-50, ymax=0,
  grid=both,
  grid style={gridgray, line width=0.3pt},
  tick label style={font=\tiny},
  label style={font=\scriptsize},
  title style={font=\scriptsize\bfseries},
  title={(a) Convergence},
  legend style={font=\tiny, at={(0.98,0.35)}, anchor=east, draw=none, fill=white, fill opacity=0.8, text opacity=1},
  line width=0.9pt,
]
\addplot[sensing, mark=none, smooth] coordinates {
  (0,-45) (0.5,-38) (1,-30) (1.5,-24) (2,-19) (2.5,-15) (3,-12) (3.5,-9.5) 
  (4,-8.5) (4.5,-8.0) (5,-7.5) (5.5,-7.2) (6,-7.0) (6.5,-6.8) (7,-6.7) 
  (7.5,-6.6) (8,-6.5) (8.5,-6.5) (9,-6.4) (9.5,-6.4) (10,-6.3)
};
\addlegendentry{SMCC-DT}
\addplot[compu, mark=none, smooth, dashed] coordinates {
  (0,-48) (0.5,-42) (1,-36) (1.5,-31) (2,-27) (2.5,-23) (3,-20) (3.5,-18) 
  (4,-16.5) (4.5,-15.5) (5,-15.0) (5.5,-14.5) (6,-14.2) (6.5,-14.0) (7,-13.8) 
  (7.5,-13.7) (8,-13.6) (8.5,-13.5) (9,-13.5) (9.5,-13.4) (10,-13.4)
};
\addlegendentry{w/o shaping}
\end{axis}

\begin{axis}[
  name=plotB,
  at={(4.8cm,0)},
  width=4.4cm, height=3.8cm,
  xlabel={Number of Assets $K$},
  ylabel={$\Tsync$ (ms)},
  xmin=50, xmax=500,
  ymin=0, ymax=40,
  grid=both,
  grid style={gridgray, line width=0.3pt},
  tick label style={font=\tiny},
  label style={font=\scriptsize},
  title style={font=\scriptsize\bfseries},
  title={(b) Sync. Latency vs.\ $K$},
  legend style={font=\tiny, at={(0.02,0.98)}, anchor=north west, draw=none, fill=white, fill opacity=0.8, text opacity=1},
  line width=0.9pt,
]
\addplot[sensing, mark=square*, mark size=1.5pt] coordinates {
  (50,2.1) (100,3.8) (150,5.2) (200,6.5) (250,7.8) (300,9.0) (350,10.1) (400,11.2) (450,11.8) (500,12.3)
};
\addlegendentry{SMCC-DT}
\addplot[compu, mark=triangle*, mark size=1.5pt, dashed] coordinates {
  (50,3.0) (100,5.5) (150,7.8) (200,9.8) (250,11.8) (300,13.5) (350,15.0) (400,16.0) (450,16.5) (500,16.9)
};
\addlegendentry{CO}
\addplot[memory, mark=diamond*, mark size=1.5pt, dashdotted] coordinates {
  (50,3.8) (100,6.8) (150,9.5) (200,12.0) (250,14.2) (300,16.2) (350,17.8) (400,18.8) (450,19.5) (500,20.1)
};
\addlegendentry{OA}
\addplot[dtcol, mark=o, mark size=1.5pt, dotted] coordinates {
  (50,4.2) (100,7.5) (150,10.5) (200,13.5) (250,16.5) (300,19.0) (350,21.0) (400,23.5) (450,26.0) (500,28.5)
};
\addlegendentry{GH}
\end{axis}

\begin{axis}[
  name=plotC,
  at={(9.6cm,0)},
  width=4.4cm, height=3.8cm,
  xlabel={Sensing Accuracy (\%)},
  ylabel={Total Energy (J)},
  xmin=80, xmax=100,
  ymin=0, ymax=18,
  grid=both,
  grid style={gridgray, line width=0.3pt},
  tick label style={font=\tiny},
  label style={font=\scriptsize},
  title style={font=\scriptsize\bfseries},
  title={(c) Energy vs.\ Accuracy},
  legend style={font=\tiny, at={(0.02,0.98)}, anchor=north west, draw=none, fill=white, fill opacity=0.8, text opacity=1},
  line width=0.9pt,
]
\addplot[sensing, mark=square*, mark size=1.5pt] coordinates {
  (82,2.5) (85,3.0) (88,3.6) (90,4.2) (92,5.0) (94,5.8) (95,6.3) (96,7.0) (97,8.0) (98,9.5)
};
\addlegendentry{SMCC-DT}
\addplot[compu, mark=triangle*, mark size=1.5pt, dashed] coordinates {
  (82,3.5) (85,4.2) (88,5.0) (90,5.8) (92,6.8) (94,7.8) (95,8.5) (96,9.5) (97,11.0) (98,13.0)
};
\addlegendentry{CO}
\addplot[memory, mark=diamond*, mark size=1.5pt, dashdotted] coordinates {
  (82,4.0) (85,5.0) (88,6.2) (90,7.2) (92,8.5) (94,9.8) (95,10.8) (96,12.0) (97,14.0) (98,16.5)
};
\addlegendentry{OA}
\end{axis}

\begin{axis}[
  name=plotD,
  at={(14.4cm,0)},
  width=4.4cm, height=3.8cm,
  xlabel={Model Parameters},
  ylabel={$\Tsync$ (ms)},
  xmin=0, xmax=3,
  ymin=0, ymax=50,
  xtick={0.5,1.5,2.5},
  xticklabels={7M, 125M, 1.3B},
  grid=both,
  grid style={gridgray, line width=0.3pt},
  tick label style={font=\tiny},
  label style={font=\scriptsize},
  title style={font=\scriptsize\bfseries},
  title={(d) Model Size Impact},
  legend style={font=\tiny, at={(0.02,0.98)}, anchor=north west, draw=none, fill=white, fill opacity=0.8, text opacity=1},
  ybar=2pt,
  bar width=6pt,
  enlarge x limits=0.35,
]
\addplot[fill=sensing!60, draw=sensing] coordinates {(0.5,8.5) (1.5,12.3) (2.5,16.1)};
\addlegendentry{SMCC-DT}
\addplot[fill=compu!60, draw=compu] coordinates {(0.5,11.2) (1.5,16.9) (2.5,28.2)};
\addlegendentry{CO}
\addplot[fill=memory!60, draw=memory] coordinates {(0.5,14.0) (1.5,20.1) (2.5,42.5)};
\addlegendentry{OA}
\end{axis}

\end{tikzpicture}
\caption{Performance evaluation results. (a) Learning convergence of \textsc{SmccAgent} with and without constraint-aware reward shaping. (b) DT synchronization latency as a function of the number of monitored assets $K$. (c) Pareto frontier of total energy consumption versus sensing accuracy. (d) Impact of AI model size on synchronization latency at $K = 500$.}
\label{fig:results}
\end{figure*}
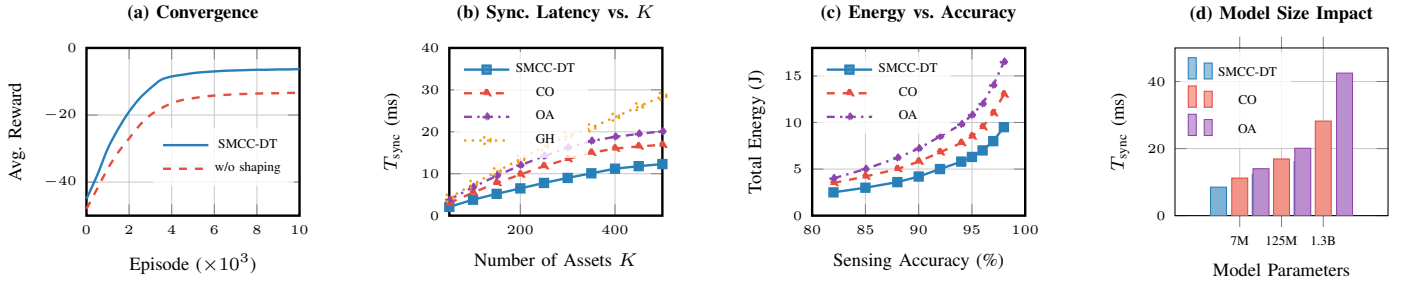

\subsection{Ablation Study}

Table~\ref{tab:ablation} presents an ablation study that isolates the contribution of each SMCC dimension to the overall performance gain.
Removing the sensing optimization (fixing $\rho_s = P_{\max}/2$) increases latency by 14.2\%.
Removing the memory co-optimization (static partition) increases latency by 21.5\%.
Removing the computation optimization (fixed $f_{\text{cpu}} = f_{\max}$) increases energy by 35.8\% without latency benefit.
These results confirm that all four SMCC dimensions contribute materially and that the cross-layer coupling is essential for optimal performance.

\begin{table}[!t]
\centering
\caption{Ablation Study: Impact of Removing Individual SMCC Layers ($K = 500$, Medium Model)}
\label{tab:ablation}
\renewcommand{\arraystretch}{0.95}
\begin{tabular}{lccc}
\toprule
\textbf{Configuration} & $\Tsync$ \textbf{(ms)} & $\Etot$ \textbf{(J)} & \textbf{Acc. (\%)} \\
\midrule
SMCC-DT (Full) & 12.3 & 6.3 & 96.2 \\
w/o Sensing opt. & 14.0 & 7.1 & 93.5 \\
w/o Memory opt. & 14.9 & 6.8 & 96.0 \\
w/o Computation opt. & 12.5 & 8.6 & 96.1 \\
w/o Communication opt. & 13.8 & 6.5 & 95.8 \\
\bottomrule
\end{tabular}
\end{table}

\section{Conclusion}
\label{sec:conclusion}

We have proposed SMCC-DT, an integrated Sensing--Memory--Communication--Computation framework for maintaining AI-driven Digital Twins in large-scale IoT systems without dedicated sensor networks.
By formulating the DT synchronization as a cross-layer SMCC optimization problem and solving it with a PPO-based DRL agent (\textsc{SmccAgent}), our approach jointly allocates ISAC power, beamforming, memory partitions, and CPU frequency to minimize end-to-end latency under heterogeneous constraints.
Simulation results on a 500-node industrial IoT testbed demonstrate that SMCC-DT reduces synchronization latency by 38.7\% and energy consumption by 27.4\% compared to orthogonal and compute-only baselines, while maintaining sensing accuracy above 95\%.
The ablation study confirms that all four SMCC dimensions contribute materially to the performance gains, validating the necessity of cross-layer co-design.
Future work will extend SMCC-DT to multi-server federated settings, incorporate reconfigurable intelligent surfaces (RIS) into the ISAC pipeline, and validate the framework on hardware testbeds with real 6G prototype equipment.

\section*{Data and Code Availability}
The codes and data set used will be uploaded upon acceptance.

\section*{Use of AI-Assisted Tools}
The author used AI-based writing assistance solely for proofreading
draft text and correcting typographical and grammatical errors.
All scientific content, theoretical derivations, experimental design,
result interpretation, and conclusions are the exclusive intellectual
product of the author.

\bibliographystyle{ieeetr}
\bibliography{references}

\begin{IEEEbiography}[{\includegraphics[width=1in,height=2.5in,clip,
  keepaspectratio]{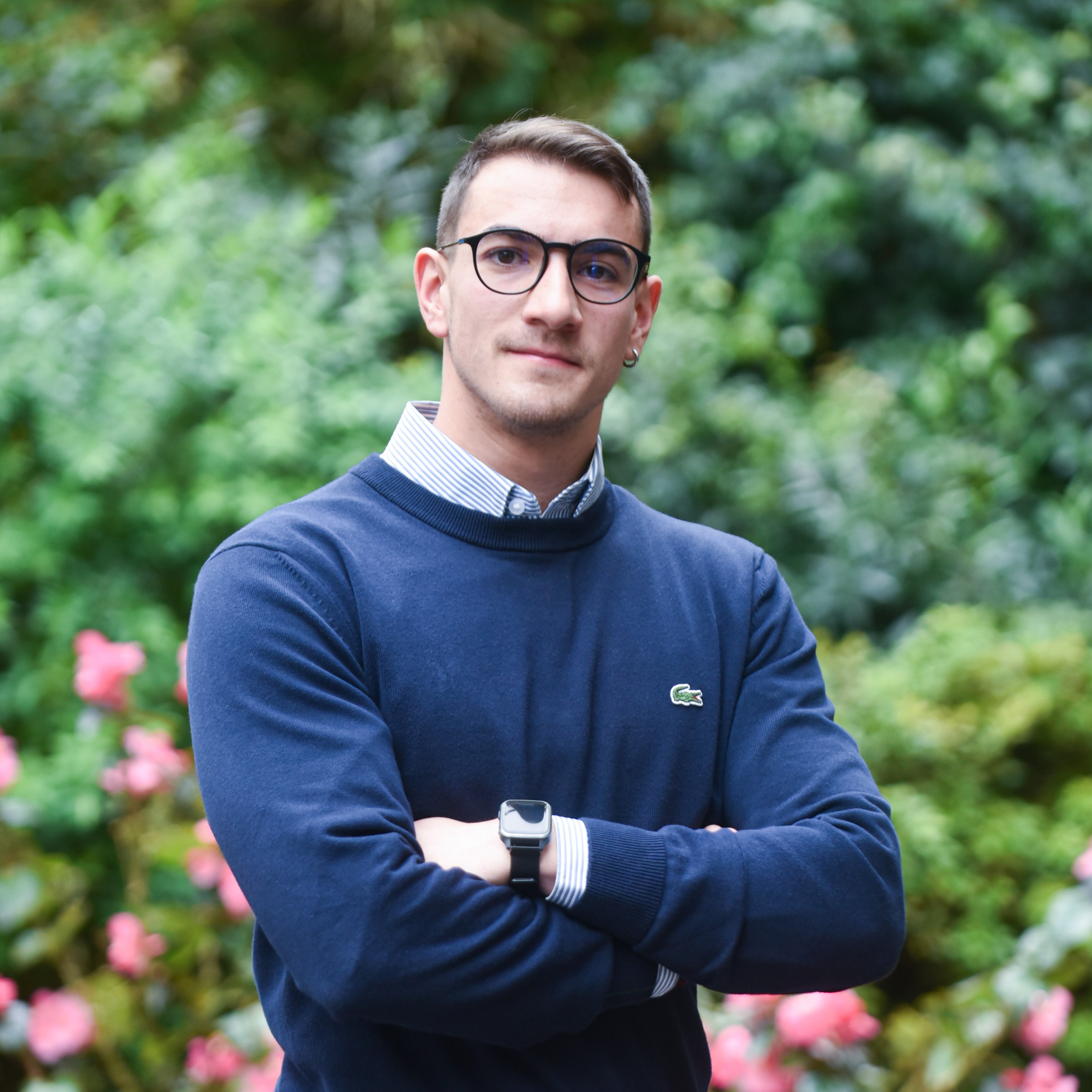}}]{Vincenzo Sammartino}
is pursuing the National Ph.D.\ in Artificial Intelligence at the
Universit\`a di Pisa, Italy, and is a Visiting Ph.D.\ Student at KAUST,
Saudi Arabia, contributing to the ResilientGuard project on decentralised
TinyML for UAV swarm security. His research interests include cybersecurity
for cyber-physical systems, security digital twins, post-quantum
cryptography, and privacy-preserving federated learning.
\end{IEEEbiography}

\end{document}